\documentclass{article}
\usepackage{a4wide}
\usepackage{amssymb,amsfonts}
\usepackage{tikz}

\let\Hung=\H
\def\H{\mathop{\hbox{\bf H}}}
\def\qed{$\square$}
\let\eps\varepsilon
\let\phi\varphi
\newtheorem{theorem}{Theorem}[section]

\newtheorem{claim}[theorem]{Claim}
\newtheorem{lemma}[theorem]{Lemma}
\newtheorem{DDDefinition}[theorem]{Definition}
\def\enddefinition{\end{DDDefinition}\egroup\medbreak}
\def\proof{\begin{demo}{Proof}}
\def\endproof{\enspace\hfill\qed\end{demo}\medbreak}
\makeatletter
\def\definition{\bgroup\def\@begintheorem##1##2{\trivlist
  \item[\hskip\labelsep{\bfseries ##1\ ##2}]}\begin{DDDefinition}}
\newcounter{scheme}\renewcommand\thescheme{\@arabic\c@scheme}
\def\fpd@scheme{tpb}\def\ftype@scheme{1}\def\ext@scheme{lof}
\def\fnum@scheme{Scheme\nobreakspace\thescheme}
\newenvironment{floatscheme}{\@float{scheme}}{\end@float}
\makeatother
\newenvironment{demo}[1]
{\par\medbreak\noindent{\bf #1\enskip}\rm\ignorespaces}{}
\def\hangitem#1{\item[\hskip-\leftmargin\enspace\textbf{#1}]\relax}
\newenvironment{hang}[1][]{\list{}{\setlength\topsep{3pt}\setlength\itemsep{0pt}\setlength\labelwidth{0pt}}\hangitem{#1}}{\endlist}
\long\def\scheme#1#2#3\endscheme{\begin{floatscheme}[htb]\begin{center}\setlength\fboxsep{1.8\fboxsep}%
\fbox{\parbox{0.9\textwidth}{\slshape #3}}\end{center}\vskip -8pt\caption{#1}\label{#2}\end{floatscheme}}

\title{\bf On-line secret sharing\thanks{This research was supported by
the ``Lend\"ulet Program'' of the Hungarian Academy of Sciences.}}
\author{L\'aszl\'o Csirmaz\thanks{Central European
University, Budapest. This research was partially supported by grant NKTH OM-00289/2008}
 \and G\'abor Tardos\thanks{School of Computing Science, Simon Fraser
   University and R\'enyi Institute, Budapest. This research was partially
   supported by NSERC Discovery grant and Hungarian OTKA grants T-046234,
   AT-048826 and NK-62321}}
\date{}

\begin{document}
\maketitle

\begin{abstract}

In a perfect secret sharing scheme the dealer distributes shares to
participants so that qualified subsets can recover the secret, while
unqualified subsets have no information on the secret. In an on-line secret
sharing scheme the dealer assigns shares in the order the participants show 
up, knowing only those qualified subsets whose all members she has seen.
We often assume that the overall access structure (the set of minimal
qualified subsets) is known and
only the order of the participants is unknown. On-line secret sharing is a
useful primitive when the set of participants grows
in time, and redistributing the secret when a new participant shows up is
too expensive. In this paper we start the investigation of unconditionally
secure on-line secret sharing schemes.

The complexity of a secret sharing scheme is the size of the largest share a
single participant can receive over the size of the secret. The infimum of
this amount in the on-line or off-line setting is the on-line or off-line
complexity of the access structure, respectively.

For paths on at most five vertices and cycles on at most six vertices the
on-line and offline complexities are equal, while for other paths and cycles
these values differ. We show that the gap between these values can be
arbitrarily large even for graph based access structures.

We present a general on-line secret sharing scheme that we call first-fit.
Its complexity is the maximal degree of the access structure. We show,
however, that this on-line scheme is never optimal: the on-line
complexity is always strictly less than the maximal degree. On the other
hand, we give examples where the first-fit scheme is almost optimal,
namely, the on-line complexity can be arbitrarily close to the maximal
degree.

The performance ratio is the ratio of the on-line and off-line complexities of
the same access structure. We show that for graphs the performance ratio is
smaller than the number of vertices, and for an infinite family of graphs
the performance ratio is at least constant times the square root of the
number of vertices.


\end{abstract}

\section{Introduction}\label{sec:intro}

Secret sharing is an important cryptographic primitive. It is used, for
example, in protocols when individual participants are either unreliable,
or participating parties don't trust each other, while they together want
to compute reliably and secretly some function of their private data. Such
protocols are, among others, electronic voting, bidding, data base access
and data base computations, distributed signatures, or joint encryptions.
Search for (efficient) secret sharing schemes led to problems in several
different branches of mathematics, and a rich theory has been developed. 
For an extended bibliography on secret sharing see \cite{stinson-wei}.

Secret sharing is a method to hide a piece of information -- the {\it
secret}
-- by splitting it up into pieces, and distributing these shares among
participants so that it can only be recovered from certain subsets of the
shares. Usually it is a trusted outsider -- the {\it dealer} -- who produces
the shares and communicates them privately to the participants. Thus to define
a secret sharing scheme we need to describe what the dealer should do.

As schemes can easily be scaled up by executing several instances
independently, the usual way to measure the efficiency of a scheme is to look
at the ratio between the size of the largest share any participant receives and
the size of the secret. The size of the shares and that of the secret is
measured by their entropy, which is roughly the minimal expected number of
bits which are necessary to define the value uniquely. We write $\H(\xi)$
to denote the Shannon entropy of the random variable $\xi$
\cite{book:csiszar}.

Let $P$ denote the set of participants. We assume that
both the secret $\xi_s$ and the share $\xi_i$ assigned to a participant $i\in
P$ are random variables distributed over a finite range and all these
variables have a joint distribution. We further require that $\H(\xi_s)>0$ to
avoid trivialities. The dealer simply draws the secret and the shares randomly
according to the given distribution, and then distributes the (random)
values of the shares to the participants.
The {\em complexity} (or worst case complexity) of the scheme $\mathcal S$,
denoted by $\sigma(\mathcal S)$ is simply the ratio between the size of the
largest share and size of the secret:
$$
  \sigma(\mathcal S) = \frac{\max_{i\in P}\H(\xi_i)}{\H(\xi_s)}.
$$
The inverse of the complexity is dubbed as the {\em rate of the scheme}, in a
strong resemblance to the decoding rate of noisy channels.

We call a hypergraph $\Gamma$ on the vertex set $P$
an {\em access structure}. A subset of the participants is {\em
qualified} if it contains a hyperedge and it is {\em unqualified} otherwise. 
We say that the secret 
sharing scheme $\mathcal S$ {\em realizes} $\Gamma$ if the values of the
shares of the participants in any qualified set uniquely determine the value of
the secret, but the shares of a set of the participants in an unqualified
subset are statistically independent of the
secret. Clearly, the non-minimal hyperedges in $\Gamma$ play no role in
defining which sets are qualified, so we can and will assume that the
hyperedges in $\Gamma$ form a {\it Sperner system} \cite{bollobas}, i.e., no
hyperedge contains another hyperedge. We further assume
that the empty set is not a hyperedge as otherwise no scheme would
realize $\Gamma$.

The {\em complexity} of
$\Gamma$ is the infimum of the complexities of all schemes realizing
$\Gamma$:
$$
     \sigma(\Gamma) = \inf \{ \sigma(\mathcal S) \,:\,
  \mathcal S \mbox{ realizes } \Gamma \},
$$
this notation was introduced in \cite{marti-padro}.
By the result of Ito {\em et al.} \cite{ito}, every non-trivial Sperner system
has a complexity, i.e., every access structure is realized by some scheme. The
complexity of their construction is the maximal degree of $\Gamma$. The
{\em degree} of a vertex in a hypergraph is the number of hyperedges
containing it. The {\em maximal degree} of $\Gamma$, denoted by $d(\Gamma)$, 
is the
maximum of the degrees of vertices of $\Gamma$. The complexity of the scheme
realizing $\Gamma$ can be reduced from $d=d(\Gamma)$ to $d - (d-1)/n$, where
$n$ is the number of participants. Another general
construction for arbitrary access structure is given by Maurer \cite{maurer}.
It is, in a sense, a dual construction, and its complexity is the maximal
number of maximal unqualified subsets a certain participant is {\it not} a
member of. Both type of constructions show that the complexity of any access
structure is at most exponential in the number $n$ of participants.
It is an open problem whether there exists an access structure with
$\sigma(\Gamma) \ge n$.

A simple observation yields that $\sigma(\Gamma)\ge1$ for
all access structures $\Gamma$ with at least one hyperedge, see, e.g.,
\cite{capocelli}. Access
structures with complexity exactly $1$ are called {\em ideal}. An intense
research was conducted to characterize ideal access structures. For example,
results in \cite{marti-padro} connect the problem of characterizing ideal
access structures to representability of certain matroids.

A widely studied special case is when all minimal qualified sets are pairs,
that is, the access structure is a graph. Stinson \cite{stinson} showed that
the complexity of a graph $G$ is at most
$(d+1)/2$ where $d$ is the maximal degree of the graph. This, together with
the lower bound in \cite{capocelli} established the complexity of
both the path and the cycle of length $n>4$ to be $3/2$. Blundo {\em et
al.}
in~\cite{tight}
showed that the $(d+1)/2$ bound is tight for certain $d$-regular graph 
families. Lower and upper
bounds on the complexity on graphs with a few nodes were investigated in
\cite{dijk}. The complexity of trees was determined in
\cite{csirmaz-tardos} to be $2-1/c$ where $c$ is the size 
of the largest
core\footnote{A {\em core} is a connected subtree such that each vertex in
the core is connected to a vertex not in the core.} in the tree. 
In particular, the complexity of every tree is strictly less than $2$.

On the other hand, based on the result of Erd\Hung os and Pyber \cite{erdos-pyber}, 
Blundo {\em et al.} \cite{BSGV} show that the
complexity of any graph $G$ on $n$ vertices is $O( n/\log n)$.
So far, however, no graph has been
found with complexity above $\Theta(\log n)$. The graph with the largest known
complexity (as a function of the number of vertices) is from \cite{d-cube},
namely the edge-graph of the $d$ dimensional hypercube. This graph is
$d$-regular, has $2^d$ vertices, and its complexity is $d/2$.

\subsection{On-line secret sharing}

In the model discussed so far the dealer generates all shares simultaneously,
and communicates them to the corresponding participants. We call such schemes
{\em off-line}. In the {\em
on-line share distribution} participants form a queue, and they receive their
shares in the order they appear. When a participant arrives
the dealer is told all those qualified subsets which are formed by
this and previously seen participants. We often assume that the dealer knows
the entire access structure at the beginning but she doesn't know the
order in which the participants arrive or the identity of the participant
when he arrives. Still she has to assign a share to him and she
cannot modify this share later. In this respect on-line secret sharing
resembles on-line graph coloring: there the color of the next vertex should be
decided knowing only that part of the graph which is spanned by this and
previous vertices.

On-line secret sharing is a useful primitive when the set of participant is not
fixed in advance and shares are assigned as participants show up. The usual
way to handle such cases is by redistributing all shares every time a new
participant shows up. Redistribution, however, has high cost, while using
on-line secret sharing can be cheap and efficient.

The study of on-line schemes stems from mainly theoretical interest. It is a
quite natural extension, and definitely more powerful than the traditional
static schemes. If we can provide a more powerful tool without giving up too
much from the efficiency, then why should we settle for less?
On the other hand, if on-line schemes turn out to be prohibitively
complicated, then we should discard them as interesting but unpractical. As
our results show, for general access structures when each participant is in
a relatively small number of qualified subsets (say, less than ten) which is
a reasonable assumption, then independently of the total number of the
participants, the complexity of the system, both on-line and off-line, has a
very low complexity (below 10). The best common bounds for the complexity of 
on-line and off-line schemes are quite close. Thus, if the access structure 
has no special properties, the efficiency of an on-line scheme is not 
inferior to that of the more restricted off-line scheme. Of course, there
are special structures -- and we will present some of them -- where the
efficiencies are quite far away from each other.

In our work we took the analogy with graph coloring. There is a
quite extensive literature for on-line graph coloring, see the bibliography
in \cite{amiller}. On-line graph coloring is interesting both for its
theoretical and practical aspects. We hope the same is, or will be, true 
for the present investigation.

The {\em on-line secret sharing} of Cachin \cite{cachin95} and follow-up
papers differ from 
our approach significantly. Cachin's model considers computationally 
secure schemes only, 
while our schemes are unconditionally secure. In addition, it requires other
authentic (but not secret) publicly accessible information, which can (or
should) be
broadcast to the participants over a public channel. In our schemes only
information possessed by the participants is necessary to recover the
secret. We are mainly interested in proving lower and upper bounds on 
the complexity of such schemes compared to the complexity of unconditionally
secure off-line
schemes, which are not touched in \cite{cachin95} at all.

Dynamic access structures were investigated by Blundo {\em et al.} in
\cite{dynamic}. Their model provides unconditional security, and the dealer
is able to activate a particular access structure out of a given collection
by sending an appropriate {\em broadcast} message to all participants. The 
dynamic is provided by the dealer's ability to choose from a range of
possible access structures, while in our on-line schemes it is the
unpredictability of the order participants appear in which makes the 
scheme dynamic.

\subsection{Our contribution}

On-line secret sharing appeared first in the
conference presentation \cite{on-line}.
In this paper we give a precise definition of this notion and define the {\em
on-line complexity} $o(\Gamma)$ of an access structure $\Gamma$ as the infimum
of the complexity of an on-line secret sharing scheme realizing it. We
present a general on-line secret sharing scheme that can realize any access
structure. We call our scheme the {\em first-fit on-line secret sharing scheme}
on account of its similarity to the simplest on-line graph coloring strategy.

\begin{theorem}\label{thm:1.1}
The on-line secret sharing scheme {\em first-fit} realizes any access structure
$\Gamma$ with complexity $d=d(\Gamma)$. In particular, $o(\Gamma)\le
d$.\end{theorem}

As usual, $P_n$ denotes the path on $n$ vertices, and $C_n$ denotes the
cycle on $n$ vertices. It is well known that the complexity of $P_n$ is
$3/2$ for $n\ge4$ and complexity of $C_n$ is also $3/2$ for $n\ge5$, see,
e.g., 
\cite{capocelli}. The following theorem separates the on-line and off-line
complexities.

\begin{theorem}\label{thm:1.2}
{\rm(i)} For paths $P_n$ with $n\le5$ and for the cycles $C_n$ with $n\le6$
the on-line and off-line complexity is the same.

\noindent
{\rm(ii)}
For paths $P_n$ with $n\ge6$ and for cycles $C_n$ with $n\ge7$
the on-line complexity is strictly above the off-line complexity.

\noindent
{\rm(iii)} The on-line complexity of both $P_n$ and $C_n$ tends to 
$2$ as $n$ tends to infinity. In fact, 
$$
    2-\frac{1}{4n}\ge o(C_{n+1})\ge o(P_n)\ge2-\frac{4}{n}.
$$
\end{theorem}

Recall that, by \cite{csirmaz-tardos}, the complexity of 
a tree is below $2$.

\begin{theorem}\label{thm:1.3}
The on-line complexities of trees is unbounded. In particular, there exists an
$n$-vertex tree $T_n$ with $o(T_n)\ge\lfloor\sqrt n\rfloor/2$. Consequently the
gap between $o(\Gamma)$ and $\sigma(\Gamma)$ can be arbitrarily large.
\end{theorem}

The {\em performance ratio} tells us how much worse the on-line scheme must be
compared to the best off-line scheme. The {\em secret sharing
performance ratio of
$\Gamma$} is defined to be $o(\Gamma)/\sigma(\Gamma)$. The similarly defined
quantity for on-line graph coloring is
sublinear in the number of vertices \cite{lovasz}, and it is at least
$n/\log^2 n$ for certain graphs with $n$ vertices \cite{amiller}. Our upper
bound on the secret sharing performance ratio of graphs comes from an upper
bound of the on-line complexity and the trivial lower bound of $1$ for the
off-line complexity:

\begin{theorem}\label{thm:1.4}
{\rm(i)} Let $d=d(G)$ be the maximal degree of the graph $G$ on $n$
vertices. Then $o(G)$, and therefore the secret sharing
performance ratio, is at most $d-1/(2dn)$.

\noindent
{\rm(ii)}
For some graphs on $n$ vertices the performance 
ratio is at least $\frac13\sqrt n$.
\end{theorem}

Finally we show that the first-fit scheme is {\em never} the best
on-line scheme. The gain, however, can be exponentially small in 
cases when minimal qualified subsets are big. Recall from Theorem
\ref{thm:1.1} that the first-fit scheme has complexity $d(\Gamma)$.

\begin{theorem}\label{thm:1.5}
Let $\Gamma$ be an access structure, $d=d(\Gamma)$ be the maximal
degree of $\Gamma$, $n$ be the number of vertices in $\Gamma$, and $r\ge 2$ be
an upper bound on the size of any hyperedge in $\Gamma$ (thus $r=2$ for
graphs). There is an on-line secret sharing scheme realizing
$\Gamma$ with complexity at most
$$ 
d - \frac 1{ndM+nd^2+n}
$$
where $M = \min (r\cdot n^{2r-3}, 3^{n-1} )$.
\end{theorem}

\subsection{Organization}
The rest of the paper is organized as follows. In section \ref{prelim}
we give precise definition for the on-line secret
sharing. In section~\ref{first-fit} we describe variants of our general
first-fit scheme and prove Theorem~\ref{thm:1.1}.
Section~\ref{path-and-circle} deals with the on-line complexity of paths and
cycles and we prove there Theorem~\ref{thm:1.2}(i). In
section~\ref{entropy} we exhibit graphs with the on-line complexity close to the maximal degree. These include
the long paths and cycles proving Theorem~\ref{thm:1.2}(ii) and trees proving
Theorem~\ref{thm:1.3}. Finally, in Section~\ref{tight} we show that the
first-fit scheme is never optimal proving Theorems~\ref{thm:1.4} and
\ref{thm:1.5}.

\section{On-line secret sharing schemes}\label{prelim}

Having defined off-line secret sharing schemes in the preceding section we
define on-line secret sharing here. On-line secret sharing relates to the secret
sharing in the same way as on-line graph coloring relates to graph
coloring. Here the structure $\Gamma$ is known in advance, and the
participants receive their
shares one by one and the assigned share cannot be changed later on.
Participants appear according to an unknown permutation. When a participant
$p$ shows up, his identity (as a vertex of $\Gamma$) is not
revealed, only those qualified subsets are shown to the dealer which $p$ is
the last member of (i.e., all other members arrived previously). 
Based only on
the emerging hypergraph (on the participants who have arrived so far) the dealer
assigns a share to the new participant. At the end the dealer will see a
permuted version of the access structure $\Gamma$ and the shares distributed
must satisfy the usual properties: the collection of shares assigned to a
qualified subset 
must determine the secret, and the collection of shares of an unqualified 
subset must be independent of the secret.

\subsection{An example}
Suppose we have three participants: $a$, $b$, $c$, and the minimal qualified
subsets are $\{a,b\}$ and $\{b,c\}$. Thus this access
structure is based on the path $P_3$ on three vertices. In Scheme
\ref{scheme:example} below we 
describe an on-line scheme realizing $P_3$.
\scheme{A sample on-line scheme}{scheme:example}
The dealer chooses two independent random bits: $r$ and $t$, and
sets the secret to be $r\oplus t$ the modulo 2 sum of these values.

When the first participant ($A$) shows up, he could be 
any of $a$, $b$, and $c$. In any case he gets the share $r$.
When the next participant shows up ($B$), the dealer also learns whether $A$
and $B$ together form a qualified set. 
\begin{hang}[If $\{A,B\}$ is independent:] the last 
participant (who did not appear yet) is $b$, and then $B$ will get the 
same share as $A$ did (that is, $r$). The last participant will receive $t$.
\hangitem{If $\{ A,B\}$ is a qualified subset:} any of $A$ and $B$
can be the middle person $b$. Nevertheless, $B$ gets the share $t$, thus
$\{A,B\}$ can recover the secret.
When the last person arrives, he is connected to either $A$ or $B$, but not
both. If he is connected to $A$, then he receives the same share as $B$ did
(that is
$t$), and if he is connected to $B$, then he receives the same share as $A$
did (that is $s$).
\end{hang}
\endscheme
As can be checked readily, qualified subsets can always recover 
the secret, and unqualified subsets have no information on the secret.
Every participant in the scheme receives a single bit no matter which order
they arrived. The secret is a single bit, thus the complexity of 
this scheme is 1.

\subsection{Formal definition}

To formalize this concept, we assume all the shares that may be assigned to
participants form a large (predetermined) finite collection of random
variables $\{\xi_\alpha\,:\,\alpha\in\Omega\}$. As usual, these and the secret
$\xi_s$ are random variables with a finite range and with a joint
distribution. We assume $\H(\xi_s)>0$. The dealer assigns one
of the variables $\xi_\alpha$ to each participant as soon as he shows up. The
choice of the index $\alpha$ for a participant depends only on the emerging
hypergraph, i.e., the set of hyperedges consisting of this and earlier
participants. In particular, assuming there is no singleton hyperedge, the
first participant always gets the same variable. Notice that the distribution
process does not depend on the values of the random variables, in fact one can
visualize the process as assigning variables to participants, and
only after all assignments evaluating the variables according to their 
joint distribution.

An on-line secret sharing scheme realizes the access structure $\Gamma$ if at
the end of the process, provided that the emerging hypergraph is indeed a
vertex-permuted copy of $\Gamma$, the shares of every qualified subset
determine the secret and the shares of every unqualified subset are independent
of the secret. Notice however that many sets of the random variables
$\xi_\alpha$ get never assigned to participants simultaneously, and those
collections do not have to satisfy any requirement.

The {\em complexity of the scheme $\mathcal S$} is the size
of the largest share divided by the size of the
secret:
$$
    \sigma(\mathcal S) = \frac{\max \{ \H(\xi_\alpha): \alpha \in\Omega\}}{\H(\xi_s)}
.
$$

The {\em on-line complexity} $o(\Gamma)$ of an access structure $\Gamma$ is the
infimum of the complexities of all on-line schemes realizing $\Gamma$:
$$
    o(\Gamma) = \inf \{ \sigma(\mathcal S) \,:\,
    \mathcal S \mbox{ is on-line and realizes } \Gamma \} .
$$
By fixing the order of the participants, any on-line scheme can be downgraded
to an off-line scheme. Consequently the on-line complexity cannot be smaller
than the off-line one: $o(\Gamma)\ge\sigma(\Gamma)$ holds for any $\Gamma$.

\section{First-fit on-line scheme}\label{first-fit}

In this section we present a general on-line secret sharing scheme. We name
it {\em first-fit scheme} because of the analogy to the first-fit on-line
graph coloring algorithm \cite{amiller}. The analogy even carries further.
As first-fit on-line coloring is oblivious of the graph structure on the
unseen vertices, similarly our
first-fit scheme works without the knowledge of the ``global'' access 
structure. However, for our scheme to work the dealer must know the maximum
degree $d$. For graphs we present a version of the scheme later where the
maximum degree does not have to be known in advance. This modified scheme has
complexity $d+1$ instead of $d$ given by the first-fit scheme.

For on-line schemes we distinguish the hyperedges containing a participant $p$
as {\em backward edges} and {\em forward edges} at $v$, with backward
edges being those that are revealed when $p$ arrives, and the forward edges
being those that will be revealed later.

Let us assume that $d$ is the maximal degree of an unknown access structure. The
first-fit on-line secret sharing scheme works described in the box below.
\scheme{General first-fit scheme}{scheme:first-fit}
The secret is a uniform random bit $s$. 

When a participant $p$ arrives do the following for each 
backward edge $E$ containing $p$:
\begin{hang}
For each participant $q\in E$ different from $p$ 
we select a previously unassigned (random) bit given to $q$ previously,
and assign it to the hyperedge $E$. We
also give a bit to $p$ which is also assigned to the hyperedge $E$. We choose 
this last bit in such a way
that the mod 2 sum of all bits assigned to $E$ be the secret $s$. 
\end{hang}
Finally, if the number
of backward edges at $p$ is $m<d$, then we give $d-m$ fresh uniform random 
unassigned bits to $p$ (in anticipation of the forward edges).
\endscheme

As an example, we give a more detailed description in Scheme
\ref{scheme:ff-for-P3} for the access structure $P_3$. Here the maximal degree $d$ 
is two.

\scheme{Details of first-fit scheme for $P_3$}{scheme:ff-for-P3}
The secret is a uniform random bit $s$.

When the first participant $A$ arrives, there is no backward edge, thus he 
gets two independent random bits $A_1$ and $A_2$.

When the second participant $B$ arrives, then we distinguish two cases:
\begin{hang}[$\bullet$ $\{A,B\}$ is a backward edge:]  the
(first) unused random bit from participant $A$ assigned to this edge will be
$A_1$. Thus the first bit 
$B$ receives is $A_1\oplus s$ from which $\{A,B\}$ can recover the secret. 
Also $B$ receives a new fresh random bit denoted as $B_1$. Either $A$ or $B$
can be the middle vertex of $P_3$. When the third participant $C$ arrives
then
   \begin{hang}[$\bullet$ if $\{C,A\}$ is an edge:] $C$ gets $A_2\oplus s$ plus an extra
random bit.
   \hangitem{$\bullet$ if $\{C,B\}$ is an edge:} $C$ gets $B_1\oplus s$ plus an extra
random bit.
   \end{hang}
\end{hang}

\begin{hang}[$\bullet$ $\{A,B\}$ is unqualified:] in this case there is no backward
edge, thus $B$ gets two fresh random bits: $B_1$ and $B_2$. $A$ and $B$ are
the two endpoints of the path $P_3$. 

When $C$ arrives, there will be two backward edges: $\{C,A\}$, and $\{C,B\}$.
The first backward edge is assigned the random bit $A_1$, and $C$ receives
$A_1\oplus s$. The second backward edge $\{C,B\}$ gets $B_1$, and $C$ also receives
$B_1\oplus s$, a total of two bits.
\end{hang}
\endscheme

\begin{proof} (Theorem~\ref{thm:1.1}) 
To check that the general first-fit scheme defined as Scheme \ref{scheme:first-fit}
is indeed a correct secret sharing 
scheme realizing the access structure $\Gamma$, first we note that 
if the maximal degree is $d$ then no participant runs out of 
unassigned bits.

Second, the complexity of the scheme is $d$ as
each participants receives exactly $d$ bits and the secret is a single 
uniform
bit. The participants in a hyperedge $E$ can determine the secret by adding
mod 2 the bits which were assigned to $E$. All bits received by an unqualified set
together with the secret form a set of independent random bits. So the 
first-fit scheme
realizes any access structure of maximal degree (at most) $d$.
\end{proof}

We remark that the bound given by this theorem matches the
complexity of the general off-line secret sharing scheme of Ito {\em et al.}
\cite{ito}. 


For graphs there is a 
modified version of the first-fit scheme \ref{scheme:first-fit} detailed in
Scheme \ref{scheme:ff-for-graphs}.
The secret is still a uniform random bit $s$,
but each participant receives a share whose size is only one more than the 
number of backward
edges containing that vertex, i.e., edges which are revealed when the vertex arrives. 
\scheme{Special first-fit scheme for graphs}{scheme:ff-for-graphs}
The secret is a uniform random bit $s$, the access structure is a graph $G$.

When the participant $p$ arrives, the dealer gives him a new random bit $r_p$
(independently from every other bits), and do the following for each 
backward edge $(q,p)\in G$:
\begin{hang}
take the random bit $r_q$ assigned to $q$, and give $r_q\oplus s$ to $p$ as
well.
\end{hang}
\endscheme
Thus
the maximum possible share size is $d+1$, slightly worse than the $d$ above.
The advantage of this modified scheme is that the dealer needs not
to know the maximum degree $d$ in advance.
It is easy to check that this scheme realizes any graph. It is interesting to
note however, that we could not find any analogous scheme for general
hypergraphs.

Yet another version of the first-fit scheme for graphs is when in 
Scheme \ref{scheme:ff-for-graphs}
we simply do not give the new random bit $r_p$ to $p$ whenever $p$
has the maximum number $d$ of backward edges. For this scheme to work
we need to know $d$ in advance. The
advantage compared to the general first-fit scheme \ref{scheme:first-fit} is 
that most participants
receive fewer than $d$ bits, only participants with $d$ or
$d-1$ backward edges receive a $d$ bit share.

\section{Paths and cycles}\label{path-and-circle}

There are cases when the on-line and off-line complexity coincide. The
simplest ones are covered by the Claim \ref{claim:threshold}. To state it we
need some definition.
Let $\Gamma$ be a hypergraph
and $S$ be a subset of the vertices of $\Gamma$. The {\em sub-hypergraph of
$\Gamma$ induced {\em(or spanned)} by $S$} is the hypergraph with vertex set $S$ and with
those hyperedges of $\Gamma$ that are contained in $S$. For simplicity we call
induced sub-hypergraphs {\em substructures}. We call a hypergraph
$\Gamma$ {\em fully symmetric} if each isomorphism between two of its
substructures can be extended to an automorphism of $\Gamma$.

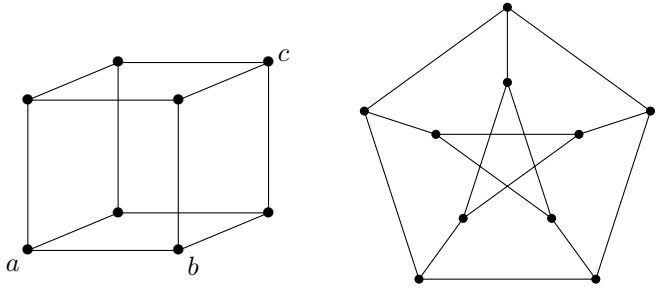
\begin{figure}
\begin{center}
\begin{tikzpicture}[scale=2.0]
\draw (0,0)node{$\bullet$}-- (0,1)node{$\bullet$} 
    -- (1,1)node{$\bullet$} -- (1,0)node{$\bullet$} -- cycle;
\draw (0.6,0.25)node{$\bullet$} -- (0.6,1.25)node{$\bullet$}
     -- (1.6,1.25)node{$\bullet$} -- (1.6,0.25)node{$\bullet$} -- cycle;
\draw (0,0)--(0.6,0.25); \draw (0,1)--(0.6,1.25); \draw (1,1)--(1.6,1.25);
\draw (1,0)--(1.6,0.25);
\draw (-0.1,-0.1) node {$a$};
\draw (1.1,-0.1) node {$b$};
\draw (1.7,1.3) node {$c$};
\end{tikzpicture}
\quad\quad
\begin{tikzpicture}[scale=1.0]
\draw (18:2cm) -- (90:2cm) -- (162:2cm) -- (234:2cm) --
(306:2cm) -- cycle;
\draw (18:1cm) -- (162:1cm) -- (306:1cm) -- (90:1cm) --
(234:1cm) -- cycle;
\foreach \x in {18,90,162,234,306}{
\draw (\x:1cm) -- (\x:2cm);
\draw[fill] (\x:2cm) circle (1.5pt);
\draw[fill] (\x:1cm) circle (1.5pt);
}
\end{tikzpicture}
\end{center}
\caption{Edge graph of the 3d cube and the Petersen graph}\label{fig:3cube-petersen}
\end{figure}
An easy example for a fully symmetric hypergraph is the $(n,k)$-threshold
structure. It consists of all $k$-element subsets of an $n$-element vertex
set. As all permutations of the vertex set is an automorphism of the
structure, any permutation of a subset can be extended to an automorphism of
the whole structure.

The edge graph of the 3d cube depicted on Figure \ref{fig:3cube-petersen} is
{\em not} fully transitive. To see this, the subgraph spanned on vertices
$a$, $b$, and $c$ has the automorphism which swaps $a$ and $b$ (and leaves
$c$ untouched). This automorphism cannot be extended to an automorphism of
the whole graph.

A less trivial example for a fully symmetric hypergraph is the so-called 
{\it Petersen graph} depicted on Figure \ref{fig:3cube-petersen}. This graph
is a 3-regular graph on 10 vertices with lots of symmetries.

\begin{claim}\label{claim:threshold}
For a fully symmetric access structure the on-line and off-line complexities
are equal.
\end{claim}

\begin{proof}
Suppose we have an off-line secret sharing scheme realizing a fully
symmetric access structure $\Gamma$
consisting of the shares $\xi_p$ for vertices $p$ of $\Gamma$ and $\xi_s$ for the
secret. We can use the very same variables for an on-line secret sharing
scheme as follows. We maintain an isomorphism $\alpha$ between the emerging hypergraph and
a substructure of $\Gamma$ and give the next participant $q$ the share
$\xi_{\alpha(q)}$. We keep $\xi_s$ in its role as the secret. Before the first
participant arrives $\alpha$ is empty. As $\Gamma$ is fully symmetric, whenever
a new participant arrives and the emerging hypergraph grows, we can
extend $\alpha$ to this new
vertex so that the value of $\alpha$ does not change on the older vertices and
$\alpha$ remains to be an isomorphism between the emerging hypergraph and a
substructure of $\Gamma$. At the end of the on-line process $\alpha$ becomes
an isomorphism between the full access structure and $\Gamma$. As the
off-line scheme realizes $\Gamma$, the
constraints on qualified and unqualified subsets will hold in this on-line
scheme as well.
\end{proof}

As a toy example, let us consider the above procedure when $\Gamma$ is the
fully symmetrical structure $C_4$. Let the four vertices be $a$, $b$, $c$
and $d$, and suppose that a perfect off-line secret sharing scheme assign 
the shares $\xi_a, \dots,\xi_d$ to these vertices.

When the first participant arrives, we pretend him to be $a$, and give him
the share $\xi_a$. When the second participant arrives, we learn whether he
is connected to the first one or not. If they are connected, then we pretend 
him to be $b$ and give him the share $\xi_b$, otherwise we think of him as
$C$ and assign him the share $\xi_c$. As the participants are
indeed the vertices of some $C_4$, after all of them arrives, their 
``pretended'' and their ``real'' roles form an
isomorphism between these structures which established the correctness of
the on-line scheme.

\medskip

Note that the strong symmetry requirement of Claim~\ref{claim:threshold} seems
to be necessary. The weaker assumption that the automorphism group of $\Gamma$
is {\em transitive} on the vertices and/or on the hyperedges is not enough. As
a counterexample, consider $C_n$, the cycle on $n\ge7$ vertices. Its
automorphism group is transitive on both the edges and vertices, but it is
not transitive on certain isomorphism classes of induced substructures. For
example no automorphism brings a pair of second neighbors to a pair of third
neighbors, despite the fact that they induce isomorphic (empty) subgraphs. The
off-line complexity of $C_n$ is $3/2$, but the on-line complexity is strictly
larger than this value (and approaches $2$ as $n$ goes to infinity) 
by Theorem~\ref{thm:1.2}.

Let $\Gamma'$ be a hypergraph obtained from $\Gamma$ by replacing
each vertex of $\Gamma$ by a nonempty class of equivalent vertices, and
replacing each hyperedge with the complete multipartite hypergraph on the
corresponding classes. We call $\Gamma'$ a {\em blowup} of $\Gamma$. Note that
$\sigma(\Gamma')=\sigma(\Gamma)$ since one can assign the same random variable
to all equivalent vertices in a class. We shall see later that the on-line
complexity of the blowup can be larger than that of
$\Gamma$. Indeed, Lemma~\ref{star} implies that the blowups of the
simple graph $G_0$ with three vertices and a single edge have unbounded
on-line complexity.

Claim~\ref{claim:threshold} applies to the {\em threshold structures}, these
are the complete uniform hypergraphs. Among graphs it applies to the complete
graphs and it also applies to the complete multi-partite graphs with equal
number of vertices in each class. All these access structures have complexity
1, so their on-line complexity is also 1. The same is true for arbitrary
complete multi-partite graphs (the blowups of complete graphs) as they are
induced subgraphs of some fully symmetric complete multipartite graph.

With these preliminaries, we turn to the complexity of paths and cycles.
First we show that the on-line complexity of short paths and cycles are the
same as their off-line complexity.

\begin{proof} (Theorem~\ref{thm:1.2}(i))
$P_2$, and $C_3$ are complete graphs, $P_3$ and $C_4$ are
complete bipartite graphs, so their on-line and off-line complexity are
the same and equal to 1.
$C_5$ is neither complete, nor complete bipartite graph, but it is fully
symmetric. So its on-line and off-line complexities agree by
Claim~\ref{claim:threshold}. $P_4$ is not fully symmetric,
still its on-line and off-line complexities are both $3/2$. To see this notice
that $P_4$ is an induced subgraph of $C_5$, so we have $o(P_4)\le
o(C_5)=\sigma(C_5)$ and it is well known that $\sigma(P_4)=\sigma(C_5)=3/2$,
see e.g., \cite{capocelli}.
A similar argument shows that $o(P_5)=\sigma(P_5)=o(C_6)=\sigma(C_6)=3/2$
once we show the bound $o(C_6)\le3/2$. We show this by
presenting an on-line secret sharing scheme of complexity $3/2$ realizing $C_6$.

\scheme{Optimal on-line scheme for $C_6$}{scheme:C6}
This scheme uses random bits $a$, $b$, $c$, $d$, $e$, $f$ and
$x$, $y$, $z$ whose joint distribution is uniform on the values satisfying
$a+b+c+d+e+f=x+y+z=0$. Here and in the list below summation is
understood modulo $2$. The random variables representing the shares and the
secret $\xi_s$ are as follows
{\setlength\arraycolsep{0.15em}\begin{eqnarray*}
\xi_1&=&(a,b+x,c),\\
\xi_2&=&(b,c+y,d),\\
\xi_3&=&(c,d+z,e),\\
\xi_4&=&(d,e+x,f),\\
\xi_5&=&(e,f+y,a),\\
\xi_6&=&(f,a+z,b),\\
\xi_7&=&(c,b+c+d+x,e+x),\\
\xi_8&=&(f+y,a+b+c+y,b),\\
\xi_s&=&(x,y,z).
\end{eqnarray*}}
Let $\Sigma$ be the cycle on the six vertices $\xi_1$, $\xi_2$, $\xi_3$, $\xi_4$,
$\xi_5$ and $\xi_6$ in this cyclic order and $\Pi$ be the cycle on the vertices
$\xi_1$, $\xi_2$, $\xi_7$, $\xi_5$, $\xi_4$ and $\xi_8$ in this cyclic
order. We assign the variables to participants such that at the end the
assignment represents an isomorphism between the emerged access structure and
either $\Sigma$ or $\Pi$. Notice that if we succeed, then the
conditions on qualified and unqualified subsets are satisfied as both
cycles $\Sigma$ and $\Pi$ represent off-line secret sharing schemes realizing $C_6$.

We start with assigning shares to participants from the intersection of
$\Sigma$ and $\Pi$ (that is, we assign one of $\xi_1$, $\xi_2$, $\xi_4$ or
$\xi_5$ until we can). We choose the variables in such a way that at any time
the assignment represents an isomorphism between the emerging graph and an
induced subgraph of the intersection. We fail when either two adjacent edges
appear in the emerging graph or three vertices form an independent set.
At that point we commit to either $\Sigma$ or $\Pi$ and assign variables so
that at the end we get an isomorphism to the selected cycle.
\endscheme

Note that in Scheme \ref{scheme:C6} the size of the secret is $\H(\xi_s)=2$, 
while the size of any share is $3$, so the complexity of the scheme is $3/2$ 
as claimed.

All neighboring pairs in $\Sigma$ and $\Pi$ can determine the secret. For
example, from $\xi_7$ and $\xi_5$ one can get $x=e+(e+x)$, then extract the 
value $b+c+d=(b+c+d+x)+x$, finally $y=(f+y)+(b+c+d)+a+e$.
We leave it to
the reader to verify that Scheme \ref{scheme:C6} works indeed for every 
permutation of the vertices.
\end{proof}

\section{The entropy method}\label{entropy}

In this section we prove lower bounds on the on-line complexity of access
structures. We start with recalling the so-called entropy method
discussed, among others, in \cite{capocelli, entropy} as that seems to be the
most powerful method for proving lower bounds for the off-line
complexity. Then we extend it to the on-line model.

Let us consider a secret
sharing scheme with the set of participants being $P$.
For any subset $A$ of $P$ we define $f(A)$ as the joint entropy
of the random variables (the shares)
belonging the members of $A$, divided by the entropy of the secret:
\begin{equation}\label{eq:entropy-method}
     f(A) = \frac{\H(\{\xi_i:i\in A\})}{\H(\xi_s)}.
\end{equation}
The so-called Shannon inequalities for the entropy, see \cite{book:csiszar}, can
be translated to linear inequalities for $f$ as follows.
\begin{enumerate}
\item[a)] $f(\emptyset)=0$,
\item[b)] monotonicity: if $A\subseteq B$ then $f(B)\ge f(A)$,
\item[c)] submodularity: $f(A)+f(B)\ge f(A\cap B)+f(A\cup B)$.
\end{enumerate}
Furthermore, if the scheme realizes an access structure $\Gamma$, then
the conditions that qualified subsets determine the secret, while unqualified
subsets are independent of it imply further inequalities:
\begin{enumerate}
\item[d)] strict monotonicity: if $A\subset B$, $A$ is unqualified but $B$
is qualified, then $f(B)\ge f(A)+1$,
\item[e)] strict submodularity: if $A$ and $B$ are both qualified but $A\cap
B$ is unqualified, then $f(A)+f(B)\ge f(A\cap B)+f(A\cup B)+1$.
\end{enumerate}
We call a real function $f$ satisfying the conditions a)--e) above an {\em
entropy function} for $\Gamma$. An entropy function $f$ is {\em
$\alpha$-bounded} if $f(A)\le\alpha$ for all singleton sets $A$. The entropy
method can be summarized as the following claim:

\begin{claim}
For any access structure $\Gamma$ there exists a $\sigma(\Gamma)$-bounded
entropy function for $\Gamma$.
\end{claim}

\begin{proof}
Let us consider a secret sharing scheme realizing $\Gamma$.
Equation~(\ref{eq:entropy-method}) defines the function $f$ and as discussed
above it is an entropy function for $\Gamma$. By the definition of
complexity it is $\alpha$-bounded for the complexity $\alpha$ of the
scheme. In case the complexity $\sigma(\Gamma)$ is not achieved as the
complexity of a scheme realizing $\Gamma$ we use a compactness argument to
finish the proof, see \cite{matus07}.
\end{proof}

The power of the entropy method lies in the fact that finding the smallest
$\alpha$ such that an $\alpha$-bounded entropy function exists for a given
$\Gamma$ is a linear programming problem and it is tractable for small access
structures. This minimal $\alpha$, denoted by $\kappa(\Gamma)$ in
\cite{marti-padro}, is a lower bound on the complexity $\sigma(\Gamma)$.

Our next theorem gives the on-line version of the entropy method. It naturally
extends to on-line complexities of {\em classes of access structures}, a
natural concept to consider, but we restrict our attention to single access
structures in this paper. Let us denote the family of
substructures of an access structure $\Gamma$ by $S(\Gamma)$.

\begin{theorem}\label{on-ent}
{\rm(i)} For every access structure $\Gamma$ there exists a system
$\{F_\Delta\,:\,\Delta\in S(\Gamma)\}$ such that $F_\Delta$ is a non-empty
collection of $o(\Gamma)$-bounded entropy functions for $\Delta$ and they
satisfy the following extension property: if $\mu$ is an isomorphism from
$\Delta_1\in S(\Gamma)$ to a substructure of $\Delta_2\in S(\Gamma)$
and $f_1\in F_{\Delta_1}$, then there exists a function $f_2\in F_{\Delta_2}$
with $f_2(\mu(A))=f_1(A)$ for any subset $A$ of the vertices in $\Delta_1$.

\noindent
{\rm(ii)} For an arbitrary substructure $\Delta$ of $\Gamma$ one has an
$o(\Gamma)$-bounded entropy function $f$ for $\Gamma$ that is symmetric on
$\Delta$, that is, for any automorphism $\mu$ of $\Delta$ one has
$f(\mu(A))=f(A)$ for all sets $A$ of the vertices of $\Delta$.
\end{theorem}

\begin{proof}
For (i) let us consider an on-line secret sharing scheme of complexity
$\alpha$ realizing $\Gamma$. For $\Delta\in S(\Gamma)$ we consider all
permutations of the vertices of $\Delta$ and the shares assigned to them
when they arrive in that order. Each assignment yields an $\alpha$-bounded
entropy function for $\Delta$ through equation~(\ref{eq:entropy-method}). We
let $F_\Delta$ be the set of these functions.

To show that the extension
property holds assume $\mu$ is an isomorphism between $\Delta_1\in S(\Gamma)$
and a substructure of $\Delta_2\in S(\Gamma)$, furthermore $f_1\in
F_{\Delta_1}$. Consider the permutation $v_1,\ldots,v_k$ of the vertices of
$\Delta_1$ yielding the entropy function $f_1$ and let $f_2$ be the entropy
function for $\Delta_2$ obtained from a permutation of its vertices starting
with $\mu(v_1),\ldots,\mu(v_k)$ followed by the rest in an arbitrary
order. After the arrival of the first $k$ vertices the situation for the
dealer is the same as when the vertices of $\Delta_1$ arrived in the given
order, so it distributes the same shares. After that she
distributes further shares, but by the definition in
(\ref{eq:entropy-method}) this will not effect the required equality
$f_2(\mu(A))=f_1(A)$ if $A$ is a set of vertices of $\Delta_1$.

This finishes the proof of part (i) in case there is an on-line secret sharing
scheme of complexity $o(\Gamma)$ for $\Gamma$. If no such scheme exists we 
should use compactness again.

For part (ii) consider the sets $F_\Delta$ and $F_\Gamma$ guaranteed by part (i)
and pick an arbitrary entropy function $f_0\in F_\Delta$. Any automorphism
$\mu$ of $\Delta$ is an isomorphism between $\Delta$ and a substructure
(namely $\Delta$) of $\Gamma$, so we have an extension $f_\mu\in F_\Gamma$ with
$f_\mu(\mu(A))=f_0(A)$ for all sets $A$ of vertices in $\Delta$. Let $f$ be the
average of these functions $f_\mu$ for the automorphisms $\mu$ of $\Delta$. It
is easy to see that the linear constraints defining an entropy function are
preserved under taking averages, so $f$ is also an entropy function for
$\Gamma$ and it is also $o(\Gamma)$-bounded like all the functions
$f_\mu$. To see that $f$ is symmetric on $\Delta$ consider an automorphism
$\mu_0$ of $\Delta$ and a set $A$ of vertices of $\Delta$ and notice that
$f(A)$ is the average of $f_\mu(A)=f_0(\mu^{-1}(A))$, while $f(\mu_0(A))$ is
the average of the same values $f_\mu(\mu_0(A))=f_0(\mu^{-1}\mu_0(A))$.
\end{proof}

Note that making an entropy function for $\Gamma$ symmetric on $\Gamma$ is
possible for off-line secret sharing schemes as well. But using
Theorem~\ref{on-ent}(ii) one can make the entropy function symmetric on a
well chosen substructure of $\Gamma$ that may have much more automorphisms
than $\Gamma$ itself.

\begin{theorem}\label{star}
Let the graph $G$ consist of a star with $d\ge2$ edges and $m$ isolated
vertices. Then
$$o(G)\ge d-\frac{d^3-d^2}{2m+2+d^2+d} > d-\frac{d^3}{2m}.$$
\end{theorem}

\begin{proof}
Let $H$ be the (empty) subgraph of $G$ spanned by all vertices but the degree
$d$ vertex $v$, the center of the star. Let
$f$ be the $o(G)$-bounded entropy function for $G$ that is symmetric on $H$,
the existence of which is claimed by Theorem~\ref{on-ent}(ii). Note
that, by symmetry, $f(A)$ is determined by $|A|$ for sets $v\notin A$, so for
such a set of size $k$ let us have $f(A)=c_k$.

Let $v_1,\ldots,v_d$ be the neighbors of $v$ and
$V_i=\{v_1,\ldots,v_i\}$. Let $H$ be an arbitrary set of isolated vertices.
By strict submodularity (rule e) for $2\le i\le d$ we have
$$f(H\cup V_{i-1}\cup\{v\})+f(H\cup \{v_i,v\})\ge f(H\cup
V_i\cup\{v\})+f(H\cup\{v\})+1.$$ 
By submodularity (rule c) for $1\le i\le d$ we have
$$f(H\cup\{v_i\}+f(H\cup\{v\})\ge f(H\cup\{v_i,v\})+f(H).$$
By rules a and c we have
$$f(H)+f(\{v\})\ge f(H\cup\{v\}),$$
and finally by strict monotonicity (rule d) we have
$$f(H\cup V_d\cup\{v\})\ge f(H\cup V_d)+1.$$
Adding all these $2d+1$ inequalities one obtains
$$\sum_{i=1}^df(H\cup\{v_i\})+f(\{v\})\ge(d-1)f(H)+f(H\cup V_d)+d.$$
All terms except $f(\{v\})$ involve subsets of $H$, so the formula simplifies
to
$$dc_{k+1}\ge(d-1)c_k+c_{k+d}+v-f(\{v\}),$$
where $k=|H|$. Introducing $\delta_i=c_{i+1}-c_i$ we can rewrite our
inequality as
$$(d-1)\delta_k\le
\delta_{k+1}+\delta_{k+2}+\cdots+\delta_{k+d-1}+d-f(\{v\}).$$
Here $k=|H|$ is arbitrary in the range $0\le k\le m$. When we add the $m+1$
corresponding inequalities most $\delta_i$ cancel. Using the bounds
$0\le\delta_i\le c_1$ (coming from monotonicity and submodularity) on the
remaining terms $\delta_i$ we obtain
$${d\choose2}c_1\ge(m+1)(d-f(\{v\})).$$
Finally as $f$ is $o(G)$-bounded we have $c_1\le o(G)$ and $f(\{v\})\le o(G)$
yielding the bound on $o(G)$ stated.
\end{proof}

We use this Theorem to prove Theorems~\ref{thm:1.2}(iii) and \ref{thm:1.3}.

\begin{proof} (Theorem~\ref{thm:1.2}(ii) and (iii))
For part (iii) notice that the graph $G$ consisting a $P_3$ component and
$\lceil n/2\rceil-2$ isolated vertices is an induced subgraph of $P_n$,
which is also an induced subgraph of $C_{n+1}$. Thus we have $o(C_{n+1})\ge
o(P_n)\ge2-4/n$, where the last inequality comes from Theorem~\ref{star}. The
upper bound on the on-line complexity of cycles comes from our general
observation that first-fit is never optimal, as stated in Theorem
\ref{thm:1.4}(i). The proof of this latter statement is  postponed to
Section~\ref{tight}.

The lower bound proved in general establishes $o(P_n)>3/2=\sigma(P_n)$ for
$n\ge9$. To find the exact threshold as claimed in part (ii) it is enough to
prove that $o(P_6)>3/2$ as the longer paths and cycles contain $P_6$ as an
induced subgraph. For this we use Theorem~\ref{on-ent}(ii) with the subgraph $H$
of $P_6$ induced by the first, second, fourth and fifth vertex of the
path. Notice that the automorphism group of $H$ has order $8$. Linear
programming shows that there is no $\alpha$-bounded entropy function on $P_6$
that is symmetric on $H$ with $\alpha<7/4$, thus the theorem tells us that
$o(P_6)\ge7/4$. In the Appendix we give a direct proof of this fact.
\end{proof}

\begin{proof} (Theorem~\ref{thm:1.3})
Consider the graph $G$ consisting of a $d$-edge star and $m$ isolated vertices
and the tree $T$ obtained by adding a vertex to $G$ and connecting it to the
center of the star and to the isolated vertices. Clearly $o(T)\ge o(G)$.
Choosing $d=\lfloor\sqrt n\rfloor$ and
$m=n-d-2$ the tree $T=T_n$ has $n$ vertices and
Theorem~\ref{star} gives the claimed lower bound on its on-line complexity.
\end{proof}

\section{Not so tight bounds on on-line complexity}\label{tight}

Stinson proved in \cite{stinson} that the 
(worst
case) complexity of any graph is at most $(d+1)/2$ where $d$ is the maximal
degree. This bound was proved to be almost sharp by van Dijk \cite{dijk}
where for
each positive $\eps$ he constructed a graph with complexity at least
$(d+1)/2-\eps$. Later Blundo {\em et al.} \cite{tight} constructed, 
for each $d\ge 2$, an
infinite family of $d$-regular graphs with exact complexity
$(d+1)/2$.

Theorem \ref{thm:1.1} claims that the on-line
complexity
is at most $d$ for a degree $d$ graph, and from Theorem
\ref{star} it follows that this bound is {\em almost tight}, 
namely, for each positive $\eps$ there is a $d$-regular graph with on-line
complexity at least $d-\eps$. In fact, the graph family defined in
\cite{tight} works here as well, as these $d$-regular graphs have no 
triangles and have arbitrarily large independent subsets. These graphs also show that the
on-line and off-line complexity can be far away, which is the conclusion of
Theorem \ref{thm:1.3}. 

In this section we show that
the bound $d$ is never sharp for on-line complexity. In other words, 
the on-line complexity of any access structure is 
always strictly less than the maximal degree. We prove this
result for graph-based structures, and only indicate how the proof can be
modified for arbitrary access structures.

The idea is that during the secret distribution we maintain
some tiny fraction of joint information among any pair of the
participants. This joint information then can be used to reduce the number
of bits the most heavily loaded participant should receive. We shall 
use a technique extending Stinson's decomposition construction from 
\cite{stinson}. 


A {\it star $k$-cover} of $G$ is a collection $\mathcal S$ of (not 
necessarily distinct) stars $\mathcal{S} = \{S_\alpha\}$ such that 
every edge of $G$ is contained in at least $k$ of the stars. The {\em weight 
of the cover $\mathcal S$}, denoted as $w(\mathcal S)$, is the maximal number 
a vertex of 
$G$ is included in some star (either as a center or as a leaf):
$$
    w(S) = \max_{v\in G} |\{ S_\alpha\in{\mathcal S} \,:\, v\in V(S_\alpha) \} | .
$$

\begin{lemma}[Stinson, {\rm \cite{stinson}}]\label{lemma:decomp}
Suppose $\mathcal S$ is a star $k$-cover of $G$. Then the complexity of $G$
is at most $w(\mathcal S)/k$.
\end{lemma}

We present the proof here because our construction will be based on it.

\begin{proof}
Let $\mathbb F$ be a large enough finite field.
We describe a secret sharing construction in which the secret is a $k$-tuple
of elements of $\mathbb F$, and each share is a collection of at most 
$w(\mathcal S)$ elements from $\mathbb F$.
Let $V$ be the $k$-dimensional vector space over $\mathbb F$.
Pick the vector $\mathbf v_\alpha \in V$ for each $S_\alpha \in
\mathcal S$ so that any $k$ of these vectors span the whole $V$. (This can 
be done if the field $\mathbb F$ has at least $|\mathcal S|$ non-zero 
elements.) The set of vectors together with their indices will be public
information, and
they do not constitute part of the secret.
The secret is a uniform random vector $\mathbf s\in V$. For each star $S_\alpha$
in the cover the dealer chooses a random element $r_\alpha\in\mathbb F$, and 
tells $r_\alpha$ (with its index) to the leaves of $S_\alpha$, and she
tells $\langle \mathbf s,\mathbf v_\alpha\rangle-r_\alpha$ to the center
of $S_\alpha$ where $\langle \mathbf s,\mathbf v_\alpha\rangle$ denotes the
inner product of these vectors.

Obviously, in this scheme every participant receives at most $w(\mathcal S)$ 
field elements. The secret consists of $k$ independent field elements
thus the complexity of the system is $w(\mathcal S)/k$, as was
claimed.

It is clear that the vertices of an edge can recover the secret: as
the edge is covered by at least $k$ stars, the two endpoints can recover 
the inner products $\langle \mathbf s,\mathbf v_\alpha\rangle$ for $k$
distinct $\alpha$. As these $\mathbf v_\alpha$ vectors span the whole
space $V$, from these inner products they can recover $\mathbf s$ as well.
On the other hand, any unqualified subset of the vertices receives field
elements that (after removing repetitions) are
independent from each other and from $\mathbf s$.
\end{proof}

Let $G$ be a graph with maximal degree $d\ge2$.
The first on-line secret sharing scheme for $G$ we describe has complexity $d$
but assigns smaller shares for most vertices. This is similar to the
modified Scheme
\ref{scheme:ff-for-graphs} presented at the end of Section~\ref{first-fit}. In
that version of first-fit all vertices receive shares of size strictly less
than the maximum of $d$ times the size of the secret except for the vertices
with $d$ or $d-1$ backward edges. In the scheme we present here only vertices
with $d$ backward edges receive maximal size shares. Recall that when a vertex
$v$ appears we categorize the incident edges as {\em backward} or {\em
forward} depending on whether the edge is revealed at that time (if it
connects $v$ to a vertex that appeared earlier) or will be revealed later.

The simplest way to apply Lemma~\ref{lemma:decomp} is to consider
the collection of stars $\{S_v\,:\, v\in V(G)\}$, where the center of $S_v$ is
$v$ and its leaves are the neighbors of $v$. Clearly, every edge of $G$
appears in exactly two of these stars and the weight of this cover is $d+1$,
so applying Lemma~\ref{lemma:decomp} one obtains $\sigma(G)\le(d+1)/2$.

Our construction can be considered as an on-line implementation of the scheme in
the proof of
Lemma~\ref{lemma:decomp}. For it to work we construct another double cover of
the edges of $G$ with stars and the corresponding shares as we go. We will
maintain that each edge appears in two stars and will have 
at most $d$ stars with center at the same vertex.
Before we start we fix a finite field $\mathbb
F$ (any field with more than $dn$ elements will do, where $n$ is the number
of vertices), and let
$V$ be a two dimensional vector space over $\mathbb F$ and choose $dn$ linearly
independent vectors ${\mathbf v}_\alpha\in V$. The secret $\mathbf s$ is a
\scheme{On-line star packing}{scheme:below-dhalf}
\def\ctr{\mathsf{center}} \def\leaf{\mathsf{leaf}}
When a vertex $v$ appears we see
its backward degree and all of its backward neighbors, but 
don't necessarily know its forward degree. Let $m$ be the number of backward
edges at $v$ and let $m'=\max(1,m)$. We assign $m'$ centers $\ctr_v^i$ for
$1\le i\le m'$ and $d$ leaves $\leaf_v^i$ for $1\le i\le d$ to $v$. For all
the centers we choose a corresponding new vector ${\mathbf v}_\alpha$. For
each backward edge $vw$ we select a distinct center assigned to $v$ and
connect it to an unused leaf at $w$. This determines the value
associated to the center selected as the value at the leaf is already
decided. We also select a distinct leaf at $v$ for each backward edge $vw$ and
connect it to $\ctr_w^1$. Here, too, this determines the value associated to
selected leaf at $v$. The remaining $d-m$ leaves at $v$ and also the one
remaining center in case $m=0$ is assigned to $v$ in anticipation of the
forward edges and are not connected to anything at this point. We select the
associated values independently and uniformly at random for each remaining
leaf or center.
\endscheme
uniform random vector from $V$. As we go we assign ``leaves''
and ``centers'' to the vertices with corresponding field elements in such a
way that elements assigned to distinct stars are independent from each other
and from the secret, all leaves of the same star receive identical elements
and the elements corresponding to the center and a leaf of a star together
determine $\langle{\mathbf s},{\mathbf v}_\alpha\rangle$ as their sum for a
distinct ${\mathbf v}_\alpha$ for each distinct star. The share of a vertex
consists of the values associated to all leaves and centers assigned to this
vertex. Clearly, if we maintain these properties, then we obtain an 
on-line scheme for $G$. Details are in Scheme \ref{scheme:below-dhalf}.

\def\zbox#1{\setbox0\hbox to 0pt{\hss$#1$\hss}\ht0=0pt\dp0=0pt\box0}
\begin{figure}[htp]
\begin{center}
\begin{tikzpicture}[scale=1.2]
\foreach \x in {0,1,2,3,5}{
\draw(\x,0)node{$\bullet$} -- (\x,5.2);
}
\draw(0,0) arc(100:80:2.836);
\draw(1,0) arc(100:80:2.836);
\draw(0,0) arc(115:65:3.517);
\draw(1,0) arc(105:75:3.732);
\draw[dashed](3,0) arc(99.46:90-9.46:6.0828);
\draw[dashed](2,0) arc(105:75:6.0);
\draw[dashed](0,0) arc(110:70:7.4);
\def\x{0.8}\def\xx{3.5}\def\xa{3.0}\def\xb{2.3}
\def\xc{3.9}\def\xd{4.4}\def\xe{5.0}
\draw(0,\x) arc(100:80:2.836);\draw[fill](0,\x) circle (2.5pt); \draw(1,\x)circle (2.5pt);
\draw(0,\x) arc(105:75:5.7);\draw(3,\x) circle (2.5pt);
\draw(0,\x) arc(110:70:7.2);\draw(5,\x) circle (2.5pt);
\draw(0,2*\x) arc(100:80:2.836);
\draw(0,2*\x) circle (2.5pt);\draw(2,2*\x) circle(2.5pt);
\draw[fill](1,2*\x) circle (2.5pt);\draw(3,2*\x)circle(2.5pt);
\draw(1,2*\x)arc(100:80:2.836);\draw(1,2*\x)arc(105:75:3.732);
\draw(1,\xb)circle(2.5pt);\draw[fill](2,\xb)circle(2.5pt);\draw(5,\xb)circle(2.5pt);
\draw(1,\xb)arc(100:80:2.836);\draw(2,\xb)arc(105:75:5.8);
\draw[fill](3,\xa)circle(2.5pt);\draw(1,\xa)circle(2.5pt);\draw(5,\xa)circle(2.5pt);
\draw(1,\xa)arc(105:75:3.732);\draw(3,\xa)arc(99.46:90-9.46:6.0828);
\draw(0,\xx)circle(2.5pt);\draw[fill](3,\xx)circle(2.5pt);
\draw(0,\xx)arc(100:80:8.70);
\draw(0,\xc)circle(2.5pt);\draw[fill](5,\xc)circle(2.5pt);
\draw(0,\xc)arc(100:80:14.2);
\draw(2,\xd)circle(2.5pt);\draw[fill](5,\xd)circle(2.5pt);
\draw(2,\xd)arc(100:80:8.7);
\draw(3,\xe)circle(2.5pt);\draw[fill](5,\xe)circle(2.5pt);
\draw(3,\xe)arc(99.46:90-9.46:6.0828);
\draw(2,\xe-0.2)circle(2.5pt);
\draw(0,-0.4)node{\zbox{x}};\draw(1,-0.4)node{\zbox{y}};
\draw(2,-0.4)node{\zbox{z}};
\draw(3,-0.4)node{\zbox{t}};\draw(5,-0.4)node{\zbox{v}};
\draw(-0.3,0.53)node{$\mathbf c^1$};
\draw(0.73,1.46)node{$\mathbf c^1$};
\draw(1.73,2.16)node{$\mathbf c^1$};
\draw(3.3,2.8)node{$\mathbf c^1$};
\draw(3.3,3.35)node{$\mathbf c^2$};
\draw(5.3,3.8)node{$\mathbf c^1$};\draw(5.3,4.3)node{$\mathbf c^2$};
\draw(5.3,4.8)node{$\mathbf c^3$};
\end{tikzpicture}
\end{center}
\vskip -15pt
\caption{Intermediate stage: next vertex $v$ is connected to $x$, $z$, and
$t$}\label{fig:coverexample}
\end{figure}
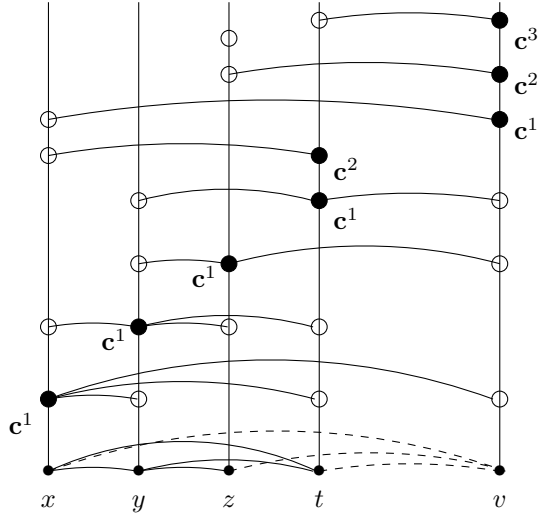

In Figure \ref{fig:coverexample} we illustrate the process for the case
when the maximal degree is $d=3$, and 
vertices $x$, $y$, $z$ and $t$ arrived previously in this order. Shares
given to participants appear as blobs on the vertical lines. 
The
assigned centers are solid dots, and the leaves are hollow ones.
When $t$ arrived, he had two backward edges going to $x$ and
$y$. Thus $t$ got two centers $\mathbf c_1$ and $\mathbf c_2$ (that is, two
indices from the pool of vectors $\mathbf v_\alpha$), and $d=3$ 
leaves. Two of the leaves were
connected to the first centers at $x$ and at $y$ respectively (and $t$ 
received the shares accordingly). 
Then $\mathbf c_1$
was connected to a free leaf of $y$ -- the corresponding share was
the $\langle \mathbf s,\mathbf v_{\alpha_1}\rangle - r_{y,3}$ difference 
where $r_{y,3}$ is the random field value assigned to $y$'s third leaf.
Similarly, and $\mathbf c_2$ was connected to the second 
free leaf of $x$, so the corresponding share was $\langle \mathbf s,
\mathbf v_{\alpha_2}\rangle - r_{x,2}$.
The third leaf of $t$ was free and had not been assigned to any star, 
thus $t$ received a fresh random element from the field as a share.

When $v$ arrives his backward degree turns out to be 3. Thus $v$ will be
assigned three centers and three leaves. The leaves are connected
to the first centers at $x$, $z$, and $t$ correspondingly, and the 
centers are connected to free leaves at those participants. The shares are
generated as above. 
Observe that each edge of the original graph
is covered exactly twice. Thus all qualified subsets can recover at least 
two inner products $\rangle\mathbf v_\alpha,\mathbf s\rangle$. As the vector space has dimension two, it means that they can also
recover the secret $\mathbf s$.

\medskip

Note that in Scheme \ref{scheme:below-dhalf} the share of participant $v$ 
consists of $m'+d$ field elements, so its size is
$(m'+d)/2$ times the size of the secret. Here $(m'+d)/2\le d-1/2$ if
$m<d$. But the complexity of the scheme is still $d$ as vertices with $d$
backward edges receive $2d$ field elements.

To apply this scheme one doesn't have to know the structure of $G$, it
is enough to know the size $n=|V(G)|$ and the maximal degree $d$. The same is
true for the more complicated scheme 
to be described below.

To push the complexity strictly
below $d$ we need to decrease the information given to vertices
of backward degree $d$ at the expense of adding further information to
all other vertices.

\begin{theorem}\label{thm:graph-tight}
Let $G$ be a graph on $n$ vertices with maximal degree $d\ge 2$. The on-line
complexity of $G$ is at most $d-1/(2dn)$.
\end{theorem}

\begin{proof}
We modify the above construction to achieve the lower complexity. Let $k$ be
a large integer to be chosen later. We execute in parallel $k$ independent
copies of the  secret distributing procedure above. Namely, the secret is $k$
independent uniform vectors ${\mathbf s}^1,\ldots,{\mathbf s}^k$ from the two
dimensional vector space $V$ and ${\mathbf s}^i$ is determined by the shares of a
qualified subset in copy $i$ of the game above. This process multiplies the
size of the secret as well as the size of the shares by $k$ so it does not
alter the complexity of the scheme. Now we modify this combined scheme as
follows.

For each pair $\{v,w\}$ of vertices if neither of them has
backward degree $d$ (but
regardless whether they form an edge or not) we assign $d$ special leaves
with identical 
random field elements independent of each other and of all other choices. We
do this by assigning $d(n-1)$ field elements to each vertex with backward
degree less than $d$ out of which values we choose $d$ -- $d$ from unused values
assigned to each earlier such vertex and we select the rest uniformly at
random. We treat these values as unused leaves.

Suppose $v$ is a vertex of backward degree $m<d$. Then we assign
$kd$ leaves and $km'$, $m'=\max(m,1)$ centers to $v$ in each copy and handle them
exactly as in scheme \ref{scheme:below-dhalf}. In addition we also assign
$d(n-1)$ special leaves to $v$ whose corresponding field elements are shared
by some other special leaf. Namely, if we have $t$ earlier vertices with
backward degree less than $d$, then we select $d$ unused special leaves
from each such vertex and make the same random value correspond to the first
$dt$ special leaves at $v$. For the remaining special leaves we select uniform
random values. In total $v$ is assigned $km'$ centers, $kd$ normal leaves and
$d(n-1)$ special leaves for a total of $km'+kd+d(n-1)$ field elements in the
share.

Next suppose $v$ is a vertex with backward degree exactly $d$. In this case we
modify just one copy of scheme \ref{scheme:below-dhalf} as follows. In the 
other $k-1$
copies of the original scheme we assign $d$ centers and $d$ leaves to $v$
each,
but in
the modified copy we assign $d$ leaves but only $d-1$ centers, one of which we
call ``special.'' All leaves and
centers are handled as in the original scheme except the special center is
connected to two leaves to take care of two backward edges. Let $x$ and $y$
be any two neighbors of $v$. We select two special leaves, one at $x$ another at
$y$ that are not participating in a star yet but which correspond to the same
random field element. We connect these special leaves to the special center at
$v$ to form a two edge star and this determines the value corresponding to the
special center. Note that this is only possible because we have leaves at
$x$ and $y$ sharing the corresponding field element as all the leaves of any
star should share the same value.
As the maximal degree is $d$, the vertex pair $\{x, y\}$ can occur at most
$d$ times in this process, thus there will always be a new pair of special
leaves to choose from. We do not assign any special leaves
to $v$ so the share of $v$ consists of $2dk-1$ elements of $\mathbb F$. 

It is clear that the scheme described is an on-line secret sharing scheme
for $G$. The secret can be written as $2k$ independent field elements.
A vertex with less than $d$ backward edges receives at most
$d(n-1)+(2d-1)k$ field elements, and a vertex with exactly $d$ backward edges 
receives $2dk-1$ field elements. Thus the complexity of the scheme is
$$
    \frac{2dk-1}{2k} = d-\frac1{2k}
$$
if $d(n-1)+(2d-1)k\le 2dk-1$, which is the case when $k=dn$. This proves the
theorem.
\end{proof}

\medskip

As the complexity of any nontrivial access structure is at least $1$, from
Theorem \ref{thm:graph-tight} it follows immediately that the performance
ratio is at most $d-1/(2dn)$ for any graph-based structure with maximal
degree $d$. This was claimed as part (i) of Theorem \ref{thm:1.4}.

\bigskip

A generalization of Theorem \ref{thm:graph-tight} for arbitrary access
structure was stated as Theorem \ref{thm:1.5}. In the construction we 
will use a bound on the number of elements in minimal qualified subsets. 
When $\Gamma$ is graph based, this bound is $2$, but in general it can 
be any number $r\le n$. As usual, $d$ denotes the maximal degree of
$\Gamma$.

Stinson's Lemma \ref{lemma:decomp}
easily generalizes for hypergraphs as follows.

\begin{lemma}
Consider a hypergraph $\Gamma$ and a system $(S_\alpha,v_\alpha)_{\alpha\in
A}$, where $S_\alpha$ is a subset of the hyperedges of $\Gamma$ and $v_\alpha$
is a vertex of $\Gamma$. Assume each hyperedge in $\Gamma$ appears in at least
$k$ of the sets $S_\alpha$ and define the weight of a vertex $x$ as
$$w(x)=\sum_{x=v_\alpha}1+\sum_{x\ne v_\alpha}|\{H\in S_\alpha\,:\,x\in
H\}|,$$
where the summations are for $\alpha\in A$. Then
$\sigma(\Gamma)\le\max_xw(x)/k$. \hfill \qed
\end{lemma}

Here we consider $S_\alpha$ as generalized stars with center $v_\alpha$. If
all hyperedges of a hypergraph $\Gamma$ have size at most $r$, then one can
consider the collection $(S_v,v)_{v\in V(G)}$, where $S_v$ is the set of
hyperedges containing $v$ together with singleton $S_\alpha$ (with an
arbitrary center) so that each hyperedge appears exactly $r$ times. The lemma
above applied to this system gives $\sigma(\Gamma)\le d-(d-1)/r$, where
$d=d(\Gamma)$ is the maximal degree. As in the case of graphs, the direct
approach to turn this scheme into an on-line scheme increases the complexity
to $d$ but only vertices with $d$ backward hyperedges receive maximal
size shares.

We only need to lower the load on participants with $d$ backward
hyperedges. Let $v$ be such a participant, and $A$ be a minimal qualified
set $v$ is in. Then $v$ gets a field element so that the sum of this
and other elements preassigned to other participants in $A$ yields the secret
value. Now $v$'s load can be lowered if he can receive the same field 
element for two different minimal qualified subsets $A_1$ and $A_2$. 
Thus we need randomly assigned numbers to $A_1-\{v\}$ and to $A_2-\{v\}$ so
that their sum be equal. Such a thing can be found if for all disjoint
subsets $U$ and $V$ of the participants with $|U|<r$, $|V|<r$ we maintain
$d$ such sums, plus $d$ further random values to be used in $(A_1\cap
A_2)-\{v\}$. These random field elements will be assigned (with appropriate
labels) to members of $U$ and $V$. 

Let $M$ be the number of the $(U,V)$ pairs a particular participant is in
either $U$ or $V$. An easy calculation shows that $M \le \min (r\cdot
n^{2r-3},3^{n-1})$. Then each participant, except for those with backward
degree $d$, will receive $d(M+d)$ extra field elements. If we execute $k$
copies of the on-line scheme in parallel, then participants with less than
$d$ backward degree receive at most $k\cdot(dn-1)+d\cdot(M+d)$ field
elements; those with exactly $d$ backward degree receive $k\cdot(dn)-1$
field elements. The secret in this case will be $kn$ field elements, thus
the complexity of the scheme is $d-1/(kn)$ if
{\setlength\arraycolsep{0.15em}
\begin{eqnarray*}
   k\cdot(dn-1)+d\cdot(M+d)&\le& k\cdot(dn) - 1\\[2pt]
    d\cdot(M+d)+1 &\le& k.
\end{eqnarray*}}
Choosing the smallest possible value for $k$ gives the complexity in Theorem
\ref{thm:1.5}.

\section{Conclusion}

In this paper we defined the notion of {on-line secret sharing scheme}, as
an extension of the classical, off-line schemes. Given a set $P$ of
participants, the dealer meets the participants one by one, and learns only
the partial structure generated by participants who show up so far. In spite
of this, final and irrevocable shares should be assigned to each
participant. The question we investigated was how much worse does an on-line
scheme perform compared to the best off-line one? 

We defined a universal on-line scheme which we called {\em first-fit} in
strong resemblance to the first-fit on-line graph coloring algorithm. Its
complexity is the maximal degree of the realized access structure, thus its
efficiency is comparable to the most efficient known general off-line
schemes.

We looked at several graph-based access structures, and found that quite
often the on-line and off-line complexities were not only close to each
other, but actually they were equal. We could separate these complexities 
by showing that
for paths on at most 5 vertices, and cycles on at most 6 vertices these
complexities are, in fact, equal. For other paths and cycles the on-line
complexity is strictly greater than the off-line (Theorem \ref{thm:1.2}. 
Nevertheless the ratio between the complexities is always less than 4/3.

For trees this {\it performance ratio} can be much larger. In fact, there is
a tree $T_n$ on $n$ vertices where this ratio is at least $\sqrt n /4$, as
proved in Theorem \ref{thm:1.3}. If the {\it maximal degree} of the access
structure is constant (say at most 10), then the situation is much better.
By Theorem \ref{thm:1.4} in this case on-line schemes are at most 10 times
more expensive than off-line ones (but very probably much less). This result
follows from the performance of the first-fit scheme.

In the last section of this paper we showed that the first-fit scheme never
optimal. For general access structure we could improve it by an exponentially
small amount only. We pose it as an open question how much this bound can be
lowered.

On-line schemes can be generated from off-line ones when the access
structure is {\it fully symmetrical}, see Claim \ref{claim:threshold}. As we
have remarked, threshold structures and complete multipartite graphs are
induced subgraphs of fully symmetrical structures, thus for them the on-line
and off-line complexities are the same. We also know that $C_5$, the cycle
on 5 points is fully symmetrical as well as the Petersen graph on Figure
\ref{fig:3cube-petersen}. An independent research problem is to
determine all fully symmetrical graphs.

Finally, we would be very much interested in computing the exact on-line 
complexity for any other graph or structure.

\section*{Acknowledgment}
The first author would like to thank Carles Padr\'o and Ronald Cramer for
their hospitality, support, and
the invitation to the RISC@CWI Conference on Combinatorics in Secret
Sharing, where the half-cooked ideas of on-line secret sharing were first
presented \cite{on-line}.

\section*{Appendix}

We give a direct proof of the fact that the on-line complexity of the path 
$P_6$ is
at least 7/4. We are using the technique discussed in Section
\ref{entropy}. Let us denote the vertices along the path $P_6$ by $a$,
$b$, $x$, $a'$, $b'$, and $y$, and let the subgraph $H$ be induced by the
vertices $a$, $b$, $a'$ and $b'$. $H$ is a matching consisting of two edges
and its automorphism group has order 8. Let $f$ be a $H$-symmetric entropy
function for $P_6$. By Theorem \ref{on-ent}(ii) it is enough to show that $f$
takes a value at least $7/4$ on some singleton.

Our starting point is the inequality
\begin{equation}\label{eqappx:2}
    f(aa'b') - f(a) \ge 3.
\end{equation}
This is well-known generalization of the inequality from \cite{BSGV}, and
follows from the fact that $a$ is not connected to any vertex of the spanned
path $xa'b'y$.

Strict submodularity and strict monotonicity yields
{\setlength\arraycolsep{0.15em}\begin{eqnarray*}
   f(bx) + f(xa') &\ge& f(ba'x) + f(x) + 1 \\
   f(ba'x) &\ge& f(ba')+1.
\end{eqnarray*}}
Using these together with $f(b)+f(x)\ge f(bx)$, $f(x)+f(a')\ge f(xa')$ we get
\begin{equation}\label{eqappx:1}
    f(b)+f(a')+f(x)\ge f(ba') + 2.
\end{equation}
As $f$ is $H$-symmetric, $f(aa')=f(ab')=f(ba')$, and
$f(a)=f(b)=f(a')=f(b')$, furthermore, by submodularity and by
(\ref{eqappx:2}),
$$
    f(aa')+f(ab') \ge f(a)+f(aa'b') \ge f(a) + (f(a)+3).
$$
Plugging this into (\ref{eqappx:1}), we get
$$
    f(b)+f(a')+f(x) \ge 2+\frac{2f(a)+3}2,
$$
from where $f(a)+f(x)\ge 7/2$. Therefore either $f(a)$ or $f(x)$ is at least
7/4, as was required.

\end{document}